\pdfoutput=1
\RequirePackage{ifpdf}
\ifpdf % We are running pdfTeX in pdf mode
\documentclass[pdftex]{sigma}
\else
\documentclass{sigma}
\fi

\usepackage{tikz}
\usetikzlibrary{arrows}
\usetikzlibrary{positioning,decorations.text}

\numberwithin{equation}{section}

\newtheorem{teor}{Theorem}[section]
\newtheorem{lem}[teor]{Lemma}
\newtheorem{RHP}[teor]{Riemann--Hilbert Problem}
{ \theoremstyle{definition}
\newtheorem{rem}[teor]{Remark} }

\newcommand{\e}{\mathrm{e}}

\renewcommand{\i}{\mathrm{i}}
\renewcommand{\d}{\mathrm{d}}
\renewcommand{\Im}{\operatorname{Im}}

\renewcommand{\(}{\left(}
\renewcommand{\)}{\right)}
\newcommand{\ol}{\overline}
\renewcommand{\Re}{\operatorname{Re}}

\begin{document}
\allowdisplaybreaks

\newcommand{\arXivNumber}{1805.05153}

\renewcommand{\PaperNumber}{119}

\FirstPageHeading

\ShortArticleName{Stimulated Raman Scattering: Solvability of Whitham System of Equations}

\ArticleName{Initial-Boundary Value Problem\\ for Stimulated Raman Scattering Model:\\ Solvability of Whitham Type System of Equations\\ Arising in Long-Time Asymptotic Analysis}

\Author{Rustem R.~AYDAGULOV~$^{\dag^1\dag^2}$ and Alexander A.~MINAKOV~$^{\dag^3\dag^4}$}

\AuthorNameForHeading{R.R.~Aydagulov and A.A.~Minakov}

\Address{$^{\dag^1}$~Lomonosov State University, Leninskie Gory 1, Moscow, Russia}
\Address{$^{\dag^2}$~A.A.~Blagonravov Institute of Mechanical Engineering, Russian Academy of Sciences,\\
\hphantom{$^{\dag^2}$}~Bardina 4, Moscow, Russia}
\EmailDD{\href{mailto:a_rust@bk.ru}{a\_rust@bk.ru}}

\Address{$^{\dag^3}$~International School for Advanced Studies (SISSA), via Bonomea 265, Trieste, Italy}
\EmailDD{\href{mailto:ominakov@sissa.it}{ominakov@sissa.it}}
\URLaddressDD{\url{https://people.sissa.it/~ominakov/}}

\Address{$^{\dag^4}$~Institut de Recherche en Math\'{e}matique et Physique (IRMP),\\
\hphantom{$^{\dag^4}$}~Universit\'{e} catholique de Louvain (UCL), Chemin du Cyclotron 2, Louvain-La-Neuve, Belgium}

\ArticleDates{Received May 15, 2018, in final form October 24, 2018; Published online November 07, 2018}

\Abstract{An initial-boundary value problem for a model of stimulated Raman scattering was considered in [Moskovchenko E.A., Kotlyarov V.P., \textit{J.~Phys.~A: Math. Theor.} \textbf{43} (2010), 055205, 31~pages]. The authors showed that in the long-time range $t\to+\infty$ the $x>0$, $t>0$ quarter plane is divided into 3 regions with qualitatively different asymptotic behavior of the solution: a region of a finite amplitude plane wave, a modulated elliptic wave region and a vanishing dispersive wave region. The asymptotics in the modulated elliptic region was studied under an implicit assumption of the solvability of the corresponding Whitham type equations. Here we establish the existence of these parameters, and thus justify the results by Moskovchenko and Kotlyarov.}


\Classification{37K15; 35Q51; 37K40}

\section{Introduction and preliminary}

In \cite{FM,Mo09,MoK06, MoK09,MoK10} the following problem was studied: consider the initial boundary value (IBV) problem for the stimulated Raman scattering model
\begin{gather}
 2\i q_t(x,t)=\mu(x,t),\qquad \mu_x(x,t)=2\i\nu(x,t)q(x,t),\nonumber\\
 \nu_x(x,t)=\i\big(\ol{q(x,t)}\mu(x,t)-\ol{\mu(x,t)}q(x,t)\big),\label{SRS}
\end{gather}
with the initial data
\begin{gather}\label{ic}
 q(x,0)=0,\qquad x\in(0,+\infty),
\end{gather}
and periodic boundary conditions
\begin{gather}\label{bc}
 \mu(0,t)=p\e^{\i\omega t},\qquad \nu(0,t)=l,
\end{gather}
where $l$, $p$, $\omega$ are real constants, which satisfy
\begin{gather*}p^2+l^2=1,\qquad -1<l<0,\qquad p>0.\end{gather*}

A reader who is interested in the physical meaning of this IBV problem can refer to~\cite{FM}. In~\cite{MoK06} the Riemann--Hilbert problem formalism for the IBV problem \eqref{SRS}--\eqref{bc} (in more general setting) was formulated, and in \cite{MoK09, MoK10} the long-time asymptotic analysis was done.

They obtained the following results, which we summarize in Theorems \ref{Theor1}, \ref{Theor2_short}, \ref{Theor3} (full formulation of their result in the region $\omega_0^2 t<x<\omega^2 t$ is too involved, and we give here a shortened version (Theorem~\ref{Theor2_short}). An interested reader can find a full formulation in the appendix (Theorem~\ref{Theor2_full})):

\begin{teor}[Kotlyarov--Mosckovchenko \cite{MoK10}, a plane wave of finite amplitude]\label{Theor1}
Let $\omega_0>0$ be defined by the formula
\begin{gather}\label{omega0}
\omega_0^2=\frac{-8l^3\omega^2}{27-18l^2-l^4+(9-l^2)\sqrt{(1-l^2)(9-l^2)}}.
\end{gather}
Then in the region $0<x<\omega_0^2 t$ the solution of the IBV problem \eqref{SRS}--\eqref{bc} for $t\to\infty$ takes the form of a plane wave
\begin{gather*} q(x,t)= \frac{-p}{2\omega}\exp\left[\i\omega t-\i\frac{l}{\omega}x-2\i\phi(\xi)\right]+\mathcal{O}\big(t^{-1/2}\big),\\
\mu(x,t)= p\exp\left[\i\omega t-\i\frac{l}{\omega}x-2\i\phi(\xi)\right]+\mathcal{O}\big(t^{-1/2}\big),\\
\nu(x,t)=1+\mathcal{O}\big(t^{-1/2}\big).
 \end{gather*}
Here \begin{gather*}\phi(\xi)=\frac{1}{2\pi}\(\int_{-\infty}^{\lambda_-(\xi)}+\int_{\lambda_+(\xi)}^{+\infty}\)\frac{\log A^2(k) \d k}{X(k)},\end{gather*}
where $\lambda_{\pm}(\xi)$ are defined in \eqref{system_dg_semisimple}, and $A(k)$, $X(k)$ are defined in \eqref{rho_def}, \eqref{X}.
\end{teor}

\begin{rem}[universality] It is claimed in \cite{MoK10} that the result of Theorem \ref{Theor1} is valid also for a more general initial datum \begin{gather*}q(x,0)=u_0(x),\end{gather*} with $u_0(x)\to 0$ sufficiently fast as $x\to+\infty$, under assumption of the absence of the discrete spectrum of the corresponding Lax pair operator. Indeed, it is natural to expect that the case of a general initial datum has the following main distinguishes with the case $u_0(x)=0$ (see for instance the case of NLS equation \cite{BM17,BIK}):
\begin{itemize}\itemsep=0pt
\item the spectral functions $A(k)$, $B(k)$ no longer have the explicit forms \eqref{rho_def}, but are constructed from the initial datum as in \cite[Section~4, formula~(31)]{MoK06}.
\item as a result, the function $\rho(k)=\frac{B(k)}{A(k)}$ is not analytic in the complex plane, and hence an additional step of its analytic approximation in the spirit of \cite{Deift1993} must be added to the series of transformations of the corresponding Riemann--Hilbert problem in~\cite{MoK10};
\item if one assumes sufficiently fast exponential decay of the initial datum $q(x,0)$ as $x\to\pm\infty$, it is standard that the function $\rho(k)$ can be extended analytically in a strip around the contour $\Sigma$. In the latter case one can skip the step with analytic approximations;
\item the functions $A(k)$, $\rho(k)$ might have degenerate behavior in the vicinity of the edge points~$E$,~$\ol{E}$ defined in~\eqref{kappa};
\item the discrete spectrum, i.e., zeros of $A(k)$, might be present. Zeros of~$A(k)$ are not necessarily simple, and correspond to solitons or breathers.
\end{itemize}
In the case when no discrete spectrum is present, the function $\rho(k)$ admits an analytic continuation in a vicinity of the contour $\Sigma$, and the behavior of $A(k)$ is not degenerate in a vicinity of the edge points, all the steps of the asymptotic analysis in~\cite{MoK10} remain without changes and hence Theorem~\ref{Theor1} holds. We see that the parameters that determine the asymptotics are determined only by the constants $l$, $p$, $\omega$, and the actual form of the initial function~$u_0(x)$ influences the leading-order asymptotic terms only as a phase shift $\phi(\xi)$.
\end{rem}

\begin{teor}[Kotlyarov--Moskovchenko \cite{MoK10}, shortened formulation, a modulated elliptic wave of finite amplitude]\label{Theor2_short}
In the region $\omega_0^2 t<x<\omega^2 t$ the solution of the IBV problem \eqref{SRS}--\eqref{bc} for $t\to\infty$ takes the form of a modulated elliptic wave
\begin{gather*}
q(x,t)=2\i\frac{\Theta_{12}(t,\xi;\infty)}{\Theta_{11}(t,\xi;\infty)}\exp\big[2\i t \widehat g_{\infty}(\xi)-2\i\widehat\phi(\xi)\big]+\mathcal{O}\big(t^{-1/2}\big),\\
 \nu(x,t)=-1+2\frac{\Theta_{11}(t,\xi;0)\Theta_{22}(t,\xi;0)}{\Theta_{11}(t,\xi;\infty)\Theta_{22}(t,\xi;\infty)}+\mathcal{O}\big(t^{-1/2}\big),\\
 \mu(x,t)=2\i\frac{\Theta_{11}(t,\xi;0)\Theta_{12}(t,\xi;0)}{\Theta_{11}^2(t,\xi;\infty)}\exp\big[2\i t \widehat g(\xi)-2\i\widehat\phi(\xi)\big]+\mathcal{O}\big(t^{-1/2}\big).
\end{gather*}
Here $\xi=\sqrt{\frac{t}{4x}}$ is a slow variable, and $\Theta_{11}(t,\xi;\infty)$, $\Theta_{12}(t,\xi;\infty)$, $\Theta_{22}(t,\xi;\infty)$, $\Theta_{11}(t,\xi;0)$, $\Theta_{12}(t,\xi;0)$, $\Theta_{22}(t,\xi;0)$, $\widehat g_{\infty}(\xi)$, $\widehat \phi(\xi)$ are some functions given explicitly in terms of the initial data.
\end{teor}
\begin{rem}[universality] It is claimed in \cite{MoK10} that the result of Theorems \ref{Theor2_short} and~\ref{Theor2_full} is true also for more general initial data $u_0(x)\to 0$ as $x\to+\infty$ under assumption of the absence of the discrete spectrum of the Lax pair operator. In the latter case all the $\Theta$-functions are determined by the parameters~$B_g$, $B_{\zeta}$ $\Delta$ (see Theorem~\ref{Theor2_full}). Further, the parameters $B_g$, $B_{\zeta}$ are fully determined solely by the constants $p$, $l$, $\omega$, and do not require knowledge of the actual form of the initial function. The only quantity that depends on the actual form of the initial function $u_0(x)$ is $\Delta$. This shows that the solution of the IBV problem behaves {\textit{universally}} with respect to the initial function.
\end{rem}
\begin{teor}[Moscovchenko \cite{Mo09}, vanishing (as $t\to\infty$) dispersive wave]\label{Theor3} In the region $x>\omega^2 t$ the solution of the IBV problem \eqref{SRS}--\eqref{bc} for $t\to\infty$ takes the form of a vanishing self-similar wave
\begin{gather*}
q(x,t)=2\sqrt{\frac{\xi^3\eta(\xi)}{t}}\exp\big[2\i\sqrt{xt}-\i\eta(\xi)\log\sqrt{x t}+\i\varphi(\xi)\big]\\
\hphantom{q(x,t)=}{} +2 \sqrt{\frac{\xi^3\eta(-\xi)}{t}}\exp\big[{-}2\i\sqrt{xt}+\i\eta(-\xi)\log\sqrt{xt}+\i\varphi(-\xi)\big]+_\mathcal{O}\big(t^{-1/2}\big) ,
\end{gather*}
where the quantity $\xi$ and functions $\eta(k)$, $\varphi(k)$ are given as follows:
\begin{gather*}\eta(k)=\frac{1}{2\pi}\log\big(1-\rho^2(k)\big),\qquad \xi=\sqrt{\frac{t}{4x}},\\
 \varphi(k)=\frac{\pi}{4}-3\eta(k)\log 2-\arg\Gamma(-\i\eta(k))+\frac{1}{\pi}\int_{-\xi}^{\xi}\log|s-k|\d\log(1-\rho^2(s)).\end{gather*}
Here $\Gamma$ is the Euler gamma-function, $\rho(k)=\frac{\varkappa^2(k)-1}{\varkappa^2(k)+1}$, and $\varkappa(k)$ is determined in~\eqref{kappa}.
\end{teor}

Theorem \ref{Theor2_short} (Theorem~\ref{Theor2_full}) was proved under an implicit assumption that the parameters of the corresponding $g$-function exist in the specified region. However, this was not proved in~\cite{MoK10}.
The goal of this paper is to establish the existence of such parameters, and thus justify the results in~\cite{MoK10}. We establish this in Theorem~\ref{Theor_main}, which is based on Lemma~\ref{lem_problem_polynomial}.

\looseness=1 We assume that the points $\d(\xi)$, $\ol{\d(\xi)}$ in equations~\eqref{lambda_-+d_system} play the role of branch points of Riemann surfaces in the Whitham modulation theory, and hence the equations~\eqref{lambda_-+d_system} are Whitham type equations (indeed, similar statement is valid in the case of the Korteweg--de Vries equation and other integrable equations (cf.~\cite{Novik}). We do not know any reference where the \mbox{genus-1} solution of SRS was previously studied, or where SRS was previously studied by Whitham theory).

There is a lot of bibliography devoted to the Whitham modulation theory approach to asymptotics of solutions of integrable equations, including cases with different finite-gap boundary conditions as $x\to\pm\infty$, see (see \cite{ABW17,Bikb3,Bikb1,Bikb4,Bikb5,Bikb6, BikN2,Bikb2, BD17,Grava16,GK07, GPT09,Novik} and the bibliography therein). Most of these results were devoted to initial-value problems associated with self-adjoint Lax operators, but there are also a few results for problems associated with non-self-adjoint Lax operators~\cite{BikN2,Novik}. We would like to mention that in our case the associated Lax operator is also non-self-adjoint, which allows possible presence of breathers and solitons of multiple order, and that the spectrum of the associated scattering problem~$\Sigma$ (defined in~\eqref{SigmaOmega}) does not lie on the real axis, which might indicate modulation instability (see \cite{BF15, K77,K17,KS99} where modulation instability for other equations is considered. We do not know a paper where this issue is considered for SRS).

In the Whitham modulation theory the question of unique solvability of Whitham equations plays the central role and the main attention is devoted to analysis of complex Whitham deformations of the corresponding Riemann surfaces. The evolution of the branch points is governed by a~transcendental system of equations, which is typically very hard to analyze, especially in the case of non-self-adjoint Lax operators~\cite{Novik}.

Recently a rigorous Riemann--Hilbert problem scheme was adjusted to the problems with non-vanishing initial data. This was achieved by introduction of the so-called $g$-function surgery approach. Equations for parameters of this $g$-function play role of Whitham equation in the Whitham modulation theory. For step-like initial data, first results were obtained by Buckingham, Venakides~\cite{BV} and independently in Boutet de Monvel, Its, Kotlyarov~\cite{BIK} for nonlinear Schr\"odinger equation (NLS). In \cite{BV} the construction of the corresponding $g$-function was done with the help of Cauchy integrals, but the system of equations for parameters of this $g$-function was too complicated to be analyzed. On the contrary, the approach in~\cite{BIK} employed the construction of $g$-function in terms of Abelian integrals in the corresponding cut complex plane. In this approach the existence of parameters for the corresponding $g$-function followed from positiveness of a polynomial of two real variables of degree~2 in a given domain.

T.~Claeys \cite{Claeys10} studied the long-time asymptotics for the solution of the Korteweg--de Vries equation with a particular choice of unbounded initial datum which grows as $\sqrt[3]{-x} $ for large $x$. In this case the solvability of the corresponding Whitham type equations was established in \cite{Potemin88}.

Further the approach in \cite{BIK} was extended to the modified Korteweg--de Vries equation (MKdV) \cite{KM_10,KM_12,KM_12_hyp,M_11a,M_11} (the degree of the corresponding polynomial was 2), and to Camassa--Holm equation \cite{M_CH_15,M_CH_16}. The step-like problem for the Korteweg--de Vries equation~\cite{EGKT} employs the same $g$-function, as MKdV. Initial value problem for NLS with more general type of initial function~\cite{BM17} also employed the same $g$-function, as in~\cite{BIK}.

In this article we show, that the Whitham type equations for the parameters of the corresponding $g$-function for the SRS model is reduced to prove positiveness of a polynomial of degree~12 in a given domain. This approach also produces a nice elementary problem (Lemma~\ref{lem_problem_polynomial}), which is however far from being trivial and can be offered to students in some mathematical olympiads. An interested reader might try to solve it first before looking in the solution.

The structure of the paper is the following: in Section~\ref{sect_RHP} we recall the Riemann--Hilbert prob\-lem formulation for the IBV problem \eqref{SRS}--\eqref{bc} from \cite{MoK06,MoK10}, and the definition of the corresponding phase $g$-functions, which are used in the asymptotic analysis of oscillatory Riemann--Hilbert prob\-lems via the nonlinear steepest descent method. In Section~\ref{sect_proof} we prove our main Theorem~\ref{Theor_main} and the underlying Lemma~\ref{lem_problem_polynomial}. It is remarkable that the values of $\alpha_0$ \eqref{alpha_0}, $x_0$~\eqref{x0} from Lemma~\ref{lem_problem_polynomial} are given for free within the framework of the Riemann--Hilbert problem analysis via the $g$-function surgery approach.

\section[Original Riemann--Hilbert problem and $g$-functions]{Original Riemann--Hilbert problem and $\boldsymbol{g}$-functions}\label{sect_RHP}
It was shown in \cite{Mo09,MoK06,MoK09,MoK10} that the solution of the IBV problem \eqref{SRS}--\eqref{bc} can be obtained from the solution of the following Riemann--Hilbert problem:

\begin{RHP}\label{RHP1}
Find a $2\times2$ matrix-valued function $M(x,t;k)$ that is sectionally analytic in $k\in\mathbb{C}\setminus\Sigma$ and satisfies the following properties:
\begin{enumerate}\itemsep=0pt
 \item[$1)$] jumps: $M_-=M_+J$, where
\begin{gather*}J(x,t;k)=\begin{cases}
 \begin{pmatrix}
 1&\rho(k)\e^{-2\i \widehat\theta(x,t;k)}\\-\rho(k)\e^{2\i \widehat\theta(x,t;k)} & 1-\rho^2(k)
 \end{pmatrix},& k\in\mathbb{R}\setminus\{0\},\\
\begin{pmatrix}
 1&0\\f(k)\e^{2\i \widehat\theta(x,t;k)}&1
\end{pmatrix},& k\in\gamma,\\
\begin{pmatrix}
 1&f(k)\e^{-2\i \widehat\theta(x,t;k)}\\0&1
\end{pmatrix},& k\in\widehat{\gamma},
 \end{cases}
\end{gather*}
\item[$2)$] end points behavior: $M(x,t;k)$ is bounded in the vicinities of the points $k=E,0,\ol E$,
\item[$3)$] asymptotics:
$M(x,t;k)=I+\mathcal{O}\big(\frac{1}{k}\big)$, $k\to\infty$.
\end{enumerate}
Here
\begin{gather}\label{kappa}\widehat\theta(x,t;k)=\frac{1}{4k}t+k x, \qquad \varkappa(k)=\sqrt[4]{\frac{k-\ol E}{k-E}},\qquad E\equiv E_1+\i E_2=\frac{l+\i p}{2\omega},\\
\label{rho_def}\rho(k)=\frac{B(k)}{A(k)},\qquad B(k)=\frac12\(\varkappa(k)-\frac{1}{\varkappa(k)}\),\qquad A(k)=\frac12\(\varkappa(k)+\frac{1}{\varkappa(k)}\),\\
\label{f_def}f(k)=\rho_-(k)-\rho_+(k)=\frac{\i}{A_-(k)A_+(k)},\end{gather}
and the contour $($see Fig.~{\rm \ref{Fig_Contour_Sigma})}
\begin{gather}\label{SigmaOmega}\Sigma\equiv \left\{k\in\mathbb{C}|\ \Im \Omega(k)=0\right\},\qquad \textrm{where}\qquad\Omega(k)=\frac{\omega}{2k}\sqrt{(k-E)(k-\ol E)}, \end{gather}
consists of the real line $\Im k=0$ and the circle arc $\gamma\cup\ol\gamma$, which is defined by equations~{\rm \cite{MoK10}}
\begin{gather*}\(k_1-\frac{|E|^2}{2E_1}\)^2+k_2^2=\(\frac{|E|^2}{2E_1}\)^2,\qquad k_1^2+k_2^2\geq|E|^2,\qquad k_1=\Re k,\qquad k_2=\Im k.\end{gather*}
The subarcs $\gamma$ and $\ol\gamma$ are symmetric with respect to the real line and are divided by it. The orientation of the contour $\Sigma$ is as follows: from $-\infty$ to $+\infty$ and from $E$ to $\ol{E}$.
\end{RHP}

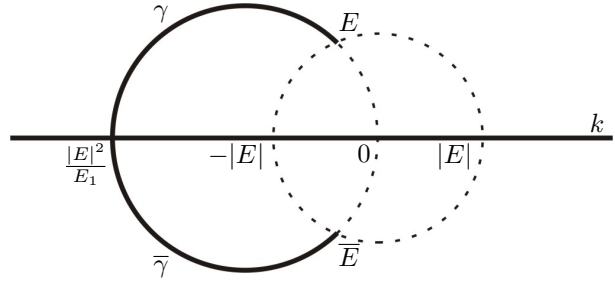
\begin{figure}[t]\centering\small
\begin{tikzpicture}
\node at (0,0) {\includegraphics[width=80mm]{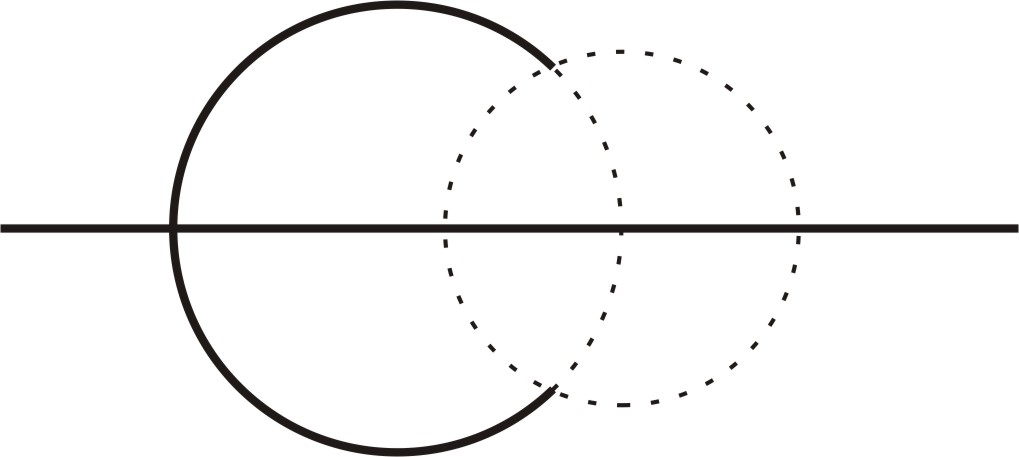}};
\node at (-3.0,-0.35) {$\frac{|E|^2}{E_1}$};
\node at (-1.0,-0.25) {$-|E|$};
\node at (0.7,-0.2) {$0$};
\node at (1.9,-0.25) {$|E|$};
\node at (-2,1.65) {$\gamma$};
\node at (-2,-1.7) {$\ol\gamma$};
\node at (0.5,1.55) {$E$};
\node at (0.5,-1.55) {$\ol{E}$};
\node at (3.8,0.2) {$k$};
\end{tikzpicture}

\caption{Contour $\Sigma$.}\label{Fig_Contour_Sigma}
\end{figure}

\begin{lem}[\cite{MoK10}] The Riemann--Hilbert Problem~{\rm \ref{RHP1}} has a unique solution and the solution of the IBV problem \eqref{SRS}--\eqref{bc} can be obtained from the solution of the Riemann--Hilbert Problem~{\rm \ref{RHP1}} in the following way:
\begin{gather*}q(x,t):=2\i\lim\limits_{k\to\infty}kM(x,t;k)_{12}=2\i\lim\limits_{k\to\infty}kM(x,t;k)_{21},\\
 Q(x,t)=\begin{pmatrix}\nu(x,t)&\i\mu(x,t)\\-\i\ol\mu(x,t)&-\nu(x,t)\end{pmatrix}:=-M(x,t;k=0)\sigma_3M^{-1}(x,t;k=0).\end{gather*}
\end{lem}

Thus, to find asymptotics of the solution of the IBV problem it is enough to find an asymptotics of the solution of the Riemann--Hilbert problem. The authors in \cite{Mo09,MoK06,MoK09,MoK10} then perform a series of transformations of this Riemann--Hilbert problem in the spirit of Deift--Zhou steepest descent method in order to reduce the problem to some explicitly solvable model problem and small-norm problems.

In turn, the crucial role in the asymptotic analysis of the Riemann--Hilbert Problem~\ref{RHP1} is played by the decomposition of the complex plane into regions where $\Im \widehat\theta\gtrless0$. In the long time asymptotics it is convenient to introduce a slow changing variable
\begin{gather*}%\label{xi}
\xi=\sqrt{\frac{t}{4x}}
\end{gather*}
and a regularized phase function \begin{gather*}\theta(k,\xi)=\frac{1}{4k}+\frac{k}{4\xi^2},\qquad \textrm{so that}\qquad \widehat\theta(x,t;k)\equiv t\theta(k,\xi).\end{gather*}

For different values of $\xi$, depending on the mutual location of the lines $\Im\theta(k,\xi)=0$ and the contour $\Sigma$, V.~Kotlyarov and A.~Moscovchenko~\cite{MoK10} introduced new phase functions $g(k,\xi)$, $\widehat g(k,\xi)$, that mimics some properties of $\theta$, such as behavior as $k\to\infty$ and $k\to0$, distribution of signs of $\Im\theta$, but also have different properties on the arcs~$\gamma$,~$\ol{\gamma}$.

The necessity of new functions $g$, $\widehat g$ appear when the curves $\Im\theta=0$ start to intersect the contour $\Sigma$ of the original Riemann--Hilbert Problem~\ref{RHP1}. Below we briefly describe the properties of $\theta$, $g$, $\widehat g$, and then prove the solvability of equations for parameters of $\widehat g$.

\subsection[Region $x>\omega^2 t$ (Zakharov--Manakov type region). A self-similar vanishing (as $t\to\infty$) wave]{Region $\boldsymbol{x>\omega^2 t}$ (Zakharov--Manakov type region).\\ A self-similar vanishing (as $\boldsymbol{t\to\infty}$) wave}
In this region the original phase function $\theta=\frac{1}{4k}+k\frac{x}{t}$ is appropriate in asymptotic analysis of the corresponding Riemann--Hilbert Problem~\ref{RHP1}. Lines where $\Im \theta=0$ are as follows~\cite{MoK10} (see Fig.~\ref{Fig_theta_signature}):
\begin{gather*}\Im\theta(k)=\frac{|k|^2-\xi^2}{4\xi^2|k|^2}\Im k.\end{gather*}

This function is appropriate in asymptotic analysis until the bold circle reaches end points $E$,~$\overline{E}$ of the arcs $\gamma$, $\ol{\gamma}$. This happens when $\xi=|E|$, i.e., $x=\omega^2 t$.

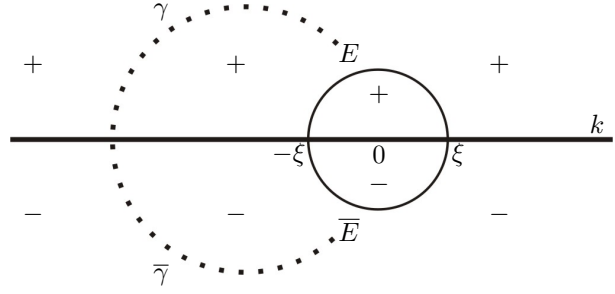
\begin{figure}[t]\centering\small
\begin{tikzpicture}
\node at (0,0) {\includegraphics[width=80mm]{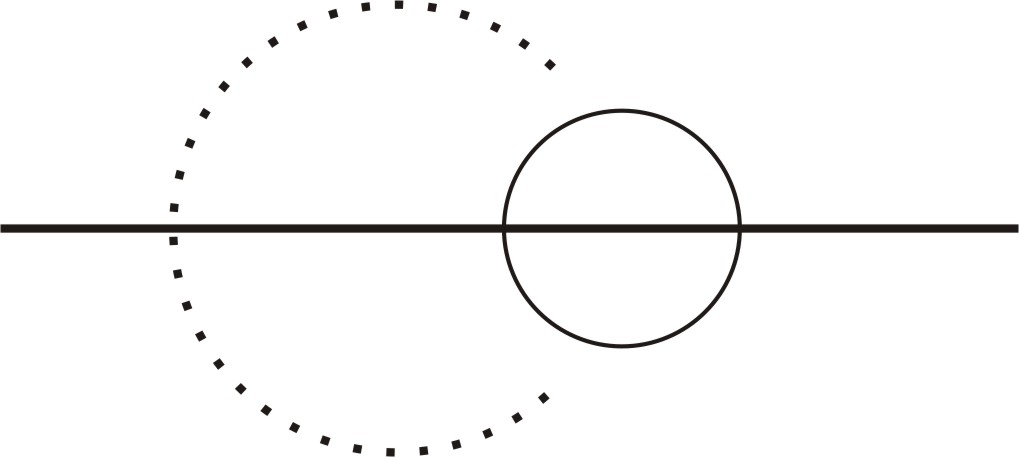}};
\node at (-2.0,1.7) {$\gamma$};
\node at (-2.0,-1.8) {$\ol\gamma$};
\node at (-3.7,1.0) {$+$};
\node at (-3.7,-1.0) {$-$};
\node at (-1.0,1.0) {$+$};
\node at (-1.0,-1.0) {$-$};
\node at (2.5,1.0) {$+$};
\node at (2.5,-1.0) {$-$};
\node at (3.8,0.2) {$k$};
\node at (0.9,0.6) {$+$};
\node at (0.9,-0.6) {$-$};
\node at (0.9,-0.2) {$0$};
\node at (-0.3,-0.2) {$-\xi$};
\node at (1.95,-0.2) {$\xi$};
\node at (0.5,1.15) {$E$};
\node at (0.5,-1.2) {$\ol{E}$};
\end{tikzpicture}

\caption{Distribution of signs of $\Im\theta(k, \xi)$ for $0<\xi<\frac{1}{2\omega}$.}\label{Fig_theta_signature}
\end{figure}

\subsection[Region $0<x<\omega_0^2 t$. A plane wave of finite amplitude]{Region $\boldsymbol{0<x<\omega_0^2 t}$. A plane wave of finite amplitude}\label{sect_g}

In this region the apropriate $g$-function surgery can be done with the help of the function
\begin{gather*}g(k)=g(k,\xi)=\Omega(k)+\frac{1}{4\xi^2}X(k)=\(\frac{\omega}{2k}+\frac{1}{4\xi^2}\)X(k),\end{gather*}
where $\Omega(k)$ is determined in \eqref{SigmaOmega} and
\begin{gather}\label{X}X(k)=\sqrt{(k-E)(k-\ol E)},\qquad \Omega(k)= \Omega(k)=\frac{\omega}{2k}\sqrt{(k-E)(k-\ol E)}.\end{gather}

\noindent To analyze how the lines $\Im g=0$ behave, it is useful to look at the differential $\d g$:\footnote{There are several misprints in \cite[p.~9]{MoK10} related to the definition of $\d g$: in the formula that follows~(21), in the second formula in~(22).}
\begin{gather*}\d g(k)=\frac{k^3-\frac{l}{2\omega} k^2+l\xi^2k-\frac{\xi^2}{2\omega}}{4\xi^2k^2X(k)}\d k=\frac{(k-\lambda_-)(k-\lambda)(k-\lambda_+)}{4\xi^2k^2X(k)}\d k,\end{gather*}
where the quantities $\lambda_-\leq\lambda\leq\lambda_+$ are subject to the following system of equations~\cite{MoK10}:
\begin{gather}\label{system_dg_semisimple} \lambda+\lambda_-+\lambda_+= \frac{l}{2\omega},\qquad \lambda(\lambda_-+\lambda_+)+\lambda_-\lambda_+=l\xi^2,\qquad \lambda\lambda_-\lambda_+= \frac{\xi^2}{2\omega}.
\end{gather}

Simple zeros of $\d g$, i.e., $\lambda_-$, $\lambda$, $\lambda_+$, correspond to the points at which there are 4 emanating rays $\Im g ={\rm const}$. Points at which $\d g\sim (k-k_0)^{-1/2}$, i.e., the points $E$, $\ol{E}$, emanate just one ray $\Im g={\rm const}$. At the point of the second order pole, i.e., at the origin, there are infinitely many lines $\Im g={\rm const}$ passing through. Lines with positive and negative value of ${\rm const}$ are separated by the real line, at which the ${\rm const}$ is~0.

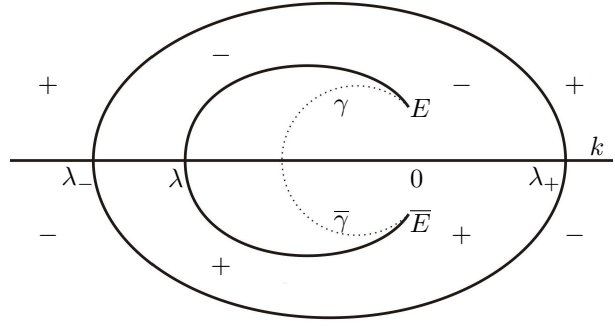
\begin{figure}[t]\centering\small
\begin{tikzpicture}
\node at (0,0) {\includegraphics[width=80mm]{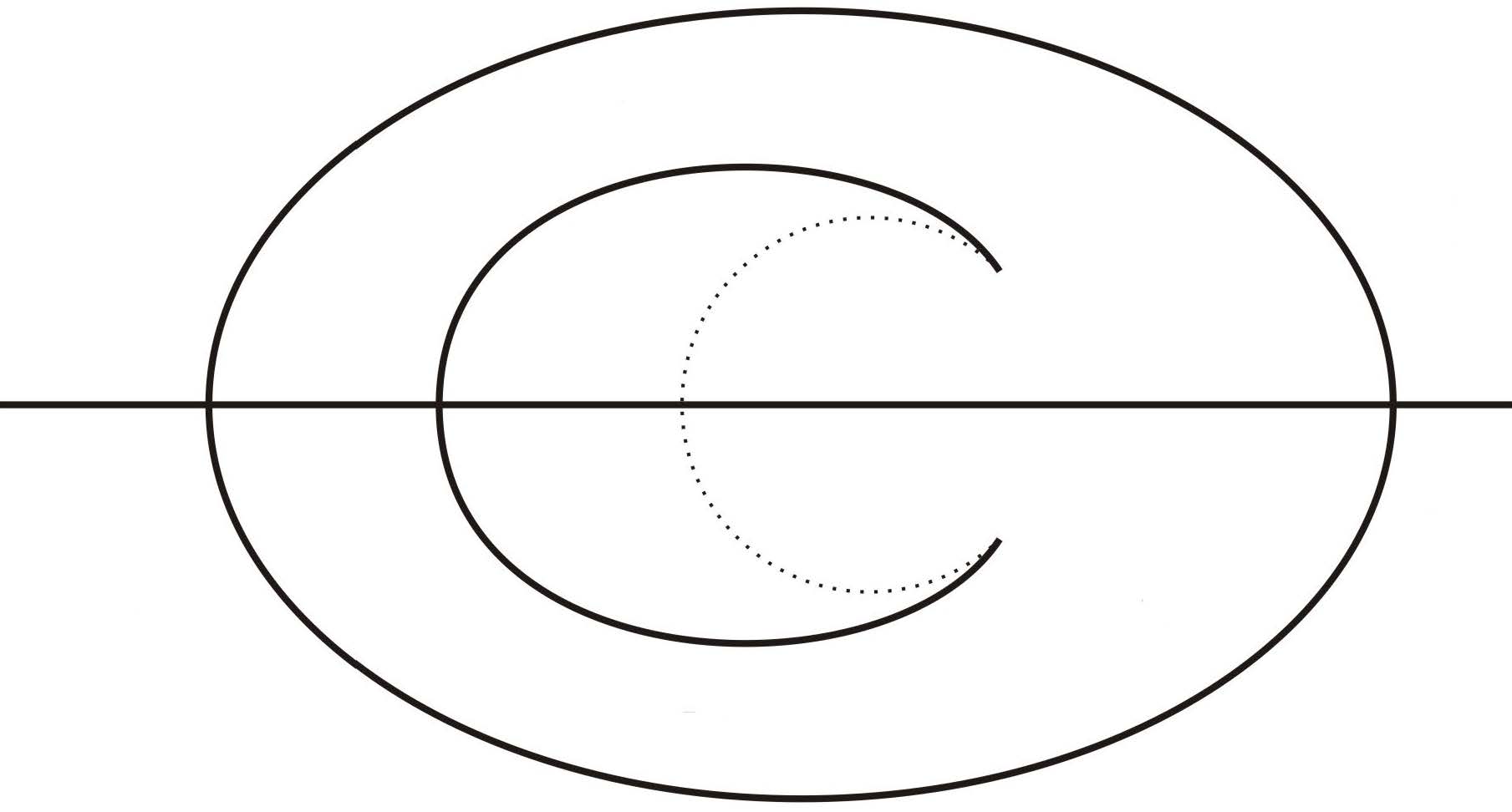}};
\node at (3.8,0.2) {$k$};
\node at (-3.1,-0.24) {$\lambda_-$};
\node at (-1.8,-0.23) {$\lambda$};
\node at (1.4,-0.23) {$0$};
\node at (3.1,-0.24) {$\lambda_+$};
\node at (0.4,0.7) {$\gamma$};
\node at (0.4,-0.8) {$\ol\gamma$};
\node at (1.45,0.7) {$E$};
\node at (1.45,-0.8) {$\ol E$};
\node at (-3.5,1.0) {$+$};
\node at (-3.5,-1.0) {$-$};
\node at (-1.2,1.4) {$-$};
\node at (-1.2,-1.4) {$+$};
\node at (2.0,1.0) {$-$};
\node at (2.0,-1.0) {$+$};
\node at (3.5,1.0) {$+$};
\node at (3.5,-1.0) {$-$};
\end{tikzpicture}

\caption{Distribution of signs of $\Im g(k, \xi)$ for $\xi>\frac{1}{2\omega_0}$. }\label{Fig_gc_signature}
\end{figure}

The system \eqref{system_dg_semisimple} was analyzed in \cite[pp.~9--10]{MoK10}, and it was shown that the real solution \begin{gather*}\lambda_-<\lambda<\lambda_+\end{gather*} exists for $\xi\in(\xi_0,+\infty)$, where $\xi_0>0$ and
\begin{gather*}%\label{xi_0}
\xi_0^2=\frac{27-18l^2-l^4+(9-l^2)\sqrt{(1-l^2)(9-l^2)}}{-32l^3\omega^2},\qquad \xi_0>0,\end{gather*}
and at the boundary of the interval $\xi=\xi_0$ we have
\begin{gather}\lambda_-(\xi_0)=\lambda(\xi_0)=\frac{3+l^2+\sqrt{(1-l^2)(9-l^2)}}{8l\omega}<0,\nonumber\\
\lambda_+(\xi_0)=\frac{3-l^2+\sqrt{(1-l^2)(9-l^2)}}{-4l\omega}>0.\label{lambda_-+boundary}\end{gather}
At the other boundary of the region, when $\xi\to+\infty$, we have
\begin{gather*}\lambda_-(\xi)\to-\infty,\qquad \lambda(\xi)\to\frac{2}{2l\omega}<0,\qquad \lambda_+(\xi)\to+\infty.\end{gather*}
Within the region $\xi_0<\xi<+\infty$ we have
\begin{gather*}-\xi\sqrt{-l}<\lambda_-(\xi)<\lambda(\xi)<\frac{1}{2l\omega}<0<\xi\sqrt{-l}<\lambda_+(\xi)<\frac{\xi}{\sqrt{-l}}.\end{gather*}

The qualitative picture of the distribution of signs of $\Im g(k,\xi)$ is plotted in Fig.~\ref{Fig_gc_signature}. This function is appropriate in the asymptotic analysis of the Riemann--Hilbert Problem~\ref{RHP1} until the points $\lambda_-$ and $\lambda$ merge.

\subsection[Region $\omega_0^2t<x<\omega^2 t$. A modulated elliptic wave of finite amplitude]{Region $\boldsymbol{\omega_0^2t<x<\omega^2 t}$. A modulated elliptic wave of finite amplitude}

In this region the appropriate $g$-function surgery can be done with the help of the function
\begin{gather}\widehat g(k)=\widehat g(k,\xi)=\int_{E}^k\d\widehat g,\nonumber\\
\d\widehat g(k)=\frac{(k-\lambda_-(\xi))(k-\lambda_+(\xi))}{4\xi^2k^2}\sqrt{\frac{(k-d(\xi))(k-\ol{d}(\xi))}{(k-E)(k-\ol{E})}}\d k.\label{widehat_g}
\end{gather}

A qualitative decomposition of the complex plane according to the distribution of signs of $\Im\widehat g(k,\xi)$ is plotted in Fig.~\ref{Fig_g_signature}. Following~\cite{MoK10}, denote the part of $\Im \widehat g=0$, which connects $E$ and $d$, by~$\gamma_d$, the part that connects $d$ and $\lambda_-$ by $\gamma_{\lambda}$, and by $\ol{\gamma_d}$, $\ol{\gamma_{\lambda}}$ the corresponding parts in the lower half-plane $\Im k\leq 0$.

\begin{figure}[t]\centering\small
\begin{tikzpicture}
\node at (0,0) {\includegraphics[width=80mm]{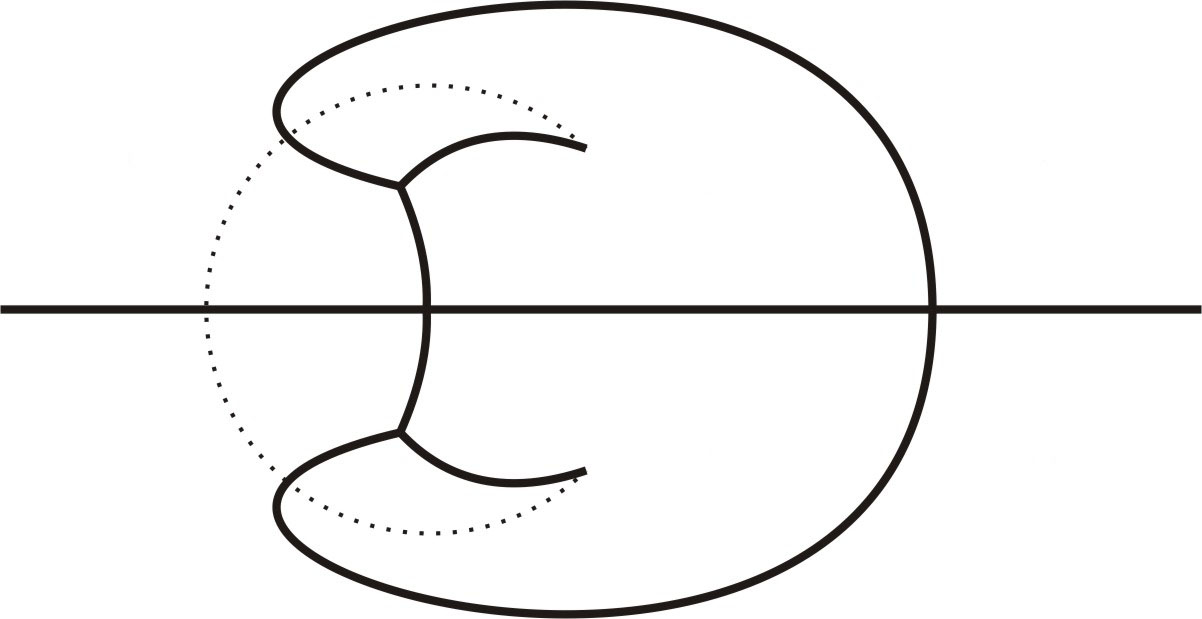}};
\node at (3.8,0.2) {$k$};
\node at (-3.5,1.0) {$+$};
\node at (-3.5,-1.0) {$-$};
\node at (3.0,1.0) {$+$};
\node at (3.0,-1.0) {$-$};
\node at (0.8,1.0) {$-$};
\node at (0.8,-1.0) {$+$};
\node at (-2.75,0.5) {$\gamma$};
\node at (-2.75,-0.5) {$\ol\gamma$};
\node at (-1.5,0.5) {$\gamma_{\lambda}$};
\node at (-1.5,-0.5) {$\ol{\gamma_{\lambda}}$};
\node at (-0.7,0.9) {$\gamma_{d}$};
\node at (-0.7,-0.9) {$\ol{\gamma_{d}}$};
\node at (-1.5,1.1) {$d$};
\node at (-1.5,-1.1) {$\ol d$};
\node at (0.1,1.1) {$E$};
\node at (0.1,-1.1) {$\ol E$};
\node at (-0.7,0.2) {$\lambda_-$};
\node at (0.3,-0.2) {$0$};
\node at (2.6,-0.3) {$\lambda_+$};
\end{tikzpicture}

\caption{Distribution of signs of $\Im \widehat g(k, \xi)$ for $\frac{1}{2\omega}<\xi<\frac{1}{2\omega_0}$. }\label{Fig_g_signature}
\end{figure}
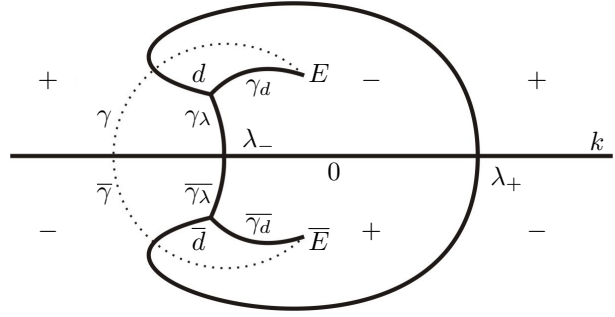

The defining properties for $\lambda_-(\xi)$, $\lambda_+(\xi)$, $d(\xi)$ are:
\begin{enumerate}\itemsep=0pt
 \item[1)] $\d \widehat g(k)-\d\theta(k)=\mathcal{O}\big(\frac{1}{k^2}\big)$, $k\to\infty$,
 \item[2)] $\widehat g(k)-\theta(k)=\mathcal{O}(1)$, $k\to 0$,
 \item[3)] $\int_{E}^{\ol{E}}\d\widehat g=0$, where the integration path does not cross the curve $\gamma_d\cup\gamma_{\lambda}\cup\ol{\gamma_{\lambda}}\cup\ol{\gamma_{d}}$.
\end{enumerate}

The condition~(3) is well-defined, i.e., the integral does not depend on the choice of contour of integration, since by the condition~(1) the residue at $\infty$ is 0.

(Let us notice that the first two of the above properties are also satisfied by~$g(k)$.)

With these properties function $\widehat g$ is analytic in $\mathbb{C}\backslash(\gamma_d\cup\gamma_{\lambda}\cup\ol{\gamma_{d}}\cup\ol{\gamma_{\lambda}})$.

It was argued in~\cite{MoK10} that under these three conditions the distribution of signs of $\Im \widehat g$ indeed will be as shown in Fig.~\ref{Fig_g_signature}. Indeed, the local structure of the lines $\Im\widehat g={\rm const}$ can be analyzed by the same reasoning as in Section~\ref{sect_g}. The only thing we need to do is to distinguish the lines where the ${\rm const}=0$.

Since $\Re\widehat g(k=E)=0$, in order to establish that $\Re\widehat g\(k=\ol{E}\)=0$ it suffices to notice that
\begin{itemize}\itemsep=0pt
\item $\widehat g(k)-\widehat g(E)\in\mathbb{R}$ for $k\in\gamma_d$, $\widehat g(k)-\widehat g\(\ol E\)\in\mathbb{R}$ for $k\in\ol\gamma_d$,
\item $\int_{\gamma_{\lambda}\cup\ol\gamma_{\lambda}}\d\widehat g \in\mathbb{R}$, where the order of integration is from $d$ to $\lambda_-$ and from $\lambda_-$ to $\ol{d}$.
\end{itemize}

The first property is due to the local analysis of the lines $\Im \widehat g={\rm const}$.
Regarding the second property, let us mention, that the property \begin{gather*}\int_{E}^{\ol{E}}\d\widehat g=0\qquad \textrm{ implies }\qquad \Im\int_{\gamma_{\lambda}}\d\widehat g =0.\end{gather*}
Indeed, we have $\int_{\gamma_{d}}\d\widehat g_{\pm}\in\mathbb{R}$, and in view of the symmetry $\ol{\d \widehat g(k)}=\d \widehat g(\ol k)$, we have
\begin{gather*}\int_{\ol{\gamma_d}}\d\widehat g_+=-\ol{\int_{\gamma_d}\d\widehat g_+}\in\mathbb{R},\qquad \int_{\ol{\gamma_{\lambda}}}\d\widehat g=-\ol{\int_{\gamma_{\lambda}}\d\widehat g},
\end{gather*}
and then from
\begin{gather*}0=\int_{E}^{\ol{E}}\d\widehat g=\(\int_{\gamma_d}\d\widehat g_++
\int_{\ol{\gamma_d}}\d\widehat g_+\)+
\(\int_{\gamma_{\lambda}}\d\widehat g+
\int_{\ol{\gamma_{\lambda}}}\d\widehat g\)\\
\hphantom{0}= \underbrace{\(\overbrace{\int_{\gamma_d}\d\widehat g_+}^{\in\mathbb{R}}-
\overbrace{\ol{\int_{{\gamma_d}}\d\widehat g_+}}^{\in\mathbb{R}}\)}_{\in\i \mathbb{R}}+
\(\int_{\gamma_{\lambda}}\d\widehat g-\ol{
\int_{{\gamma_{\lambda}}}\d\widehat g}\)
\end{gather*}
we get the required property \begin{gather*}\Im\int_{\gamma_{\lambda}}\d\widehat g=0,\qquad \Im\int_{\ol\gamma_{\lambda}}\d\widehat g=0.\end{gather*}

The three properties (1) listed above are equivalent to the system~\eqref{lambda_-+d_system}, and we prove its solvability in the next section.

\section[Proof of existence of parameters of $g$-function]{Proof of existence of parameters of $\boldsymbol{g}$-function}\label{sect_proof}

\begin{teor}\label{Theor_main}
Let the constant $E\equiv E_1+\i E_2$, $E_1<0$, $E_2>0$ be as in \eqref{kappa}, and the parame\-ter~$\xi$ be within the interval $\big(\xi_0=\frac{1}{2\omega_0}, \frac{1}{2\omega}\big)$, where $\omega_0>0$ is as in~\eqref{omega0}. Then the following system of equations
\begin{gather}
\lambda_-+\lambda_+=E_1-d_1,\qquad \lambda_-\lambda_+=-\xi^2\frac{|E|}{|d|},\nonumber\\
2\lambda_-\lambda_+d_1+(\lambda_-+\lambda_+)|d|^2=-\xi^2\(\frac{E_1}{|E|}|d|+ \frac{d_1}{|d|}|E|\),\nonumber\\
\int_{E}^{\overline{E}}\d \widehat g=0,\qquad \text{where}\qquad \d \widehat g=\frac{(k-\lambda_-)(k-\lambda_+)}{4\xi^2k^2}\sqrt{\frac{(k-d)(k-\overline{d})}{(k-E)(k-\overline{E})}}\d k,\label{lambda_-+d_system}
\end{gather}
for unknowns
\begin{gather*}\lambda_-<\lambda_+,\qquad d\equiv d_1+\i d_2,\qquad d_2>0,\end{gather*}
determines uniquely the quantities $\lambda_{\pm}$, $d\equiv d_1+\i d_2$, $d_2>0$ as functions of $\xi$.
\end{teor}

\begin{rem}In the limiting case when $\xi=\frac{1}{2\omega_0}$ (the border with the region of the plane wave asymptotics), we have $\lambda_-=d=\ol{d}$, and $\lambda_{\pm}$ are defined in \eqref{lambda_-+boundary}.

In the other limiting case when $\xi=\frac{1}{2\omega}$ (the border with the Zakharov--Manakov region) we have $d=E$, $\lambda_+=|E|$, $\lambda_-=-|E|$.
\end{rem}

\begin{proof}We have
\begin{gather*}E=\frac{l+\i p}{2\omega},\qquad l<0,\qquad l^2+p^2=1,\qquad \omega>0.\end{gather*}
Then the first two equations give us that $\lambda_\pm$ are the roots of the equation
\begin{gather*}\lambda_\pm^2-(E_1-r\cos{\varphi})\lambda_\pm-\frac{\xi^2}{2r\omega}=0,\end{gather*}
and since $\frac{-\xi^2}{2r\omega}<0$, this quadratic equation always have two distinct real solutions $\lambda_-<0<\lambda_+$.

Now the problem is reduced to that of finding $d_1$, $|d|$. Denote
\begin{gather*}d=:r(\cos{\varphi}+\i\sin{\varphi}).\end{gather*}
The third condition takes the form
\begin{gather}
-\frac{\xi^2\cos{\varphi}}{\omega}+\(\frac{l}{2\omega}-r\cos{\varphi}\)r^2= -\xi^2\(lr+\frac{\cos{\varphi}}{2\omega}\)
\quad \Leftrightarrow\nonumber\\
\cos{\varphi}\(-\frac{\xi^2}{\omega}-r^3+\frac{\xi^2}{2\omega}\)=\frac{-r^2l}{2\omega}-\xi^2lr \quad \Leftrightarrow \quad \cos{\varphi}\(-\frac{\xi^2}{2\omega}-r^3\)=\frac{-r^2l}
{2\omega}-\xi^2lr\quad \Leftrightarrow\nonumber\\
\cos{\varphi}=\frac{\frac{r^2l} {2\omega}+\xi^2lr}{\frac{\xi^2}{2\omega}+r^3}.\label{cosfi}
\end{gather}

Hence, now the problem is reduced to that of finding $r$. However, afterwards we will need to check that the module of the r.h.s.\ in~\eqref{cosfi} is less or equal than~1.

We have the connection between $r$ and $\xi$:
\begin{gather*} F(r,\xi)\equiv\int_E^{\overline{E}}\frac{(k-\lambda_-)(k-\lambda_+)}{4\xi^2k^2}\sqrt{\frac{(k-d)(k-\overline{d})} {(k-E)(k-\overline{E})}}\d k=0.\end{gather*}
The l.h.s.\ is indeed a function only of $\xi$, $r$, as was shown previously:
\begin{gather*}F(r,\xi)=\int_E^{\overline{E}}\frac{k^2-\(\frac{l}{2\omega}-r\cos{\varphi}\)k-
\frac{\xi^2}{2r\omega}}{4\xi^2k^2}\sqrt{\frac{(k-r\cos\varphi)^2+r^2(1-\cos^2\varphi)}
{(k-E)(k-\overline{E})}}\d k=0,\end{gather*}
where $\cos\varphi=(\cos\varphi)(r)$ is defined in \eqref{cosfi}.

Let us differentiate $F(r,\xi)$ with respect to $r$:
\begin{gather*}F_{r}(r,\xi)=\frac{1}{4|E|^2r^2\xi^2\big(r^3+|E|\xi^2\big)^3} \\
\hphantom{F_{r}(r,\xi)=}{} \times \int_{E}^{\overline{E}}\frac{1}{\sqrt{(k-E)(k-\overline{E})\((k-r\cos\varphi)^2+r^2(1-\cos^2\varphi)\)}}\d k\\
\hphantom{F_{r}(r,\xi)=}{} \times\Big(|E|^6\xi^8+\big(4|E|^5r^3+2E_1^2|E|^3r^3-6E_1^2|E|r^5\big)\xi^6+\big(3E_1^2|E|^4r^4-14E_1^2|E|^2r^6 \\
\hphantom{F_{r}(r,\xi)=}{} +6|E|^4r^6+3E_1^2r^8\big)\xi^4+ \big({-}6E_1^2|E|^3r^7+4|E|^3r^9+2E_1^2|E|r^9\big)\xi^2+|E|^2r^{12}\Big).
\end{gather*}

We see that vanishing or nonvanishing of $F_{r}$ depends only on the third multiplier. Let us divide there by $|E|^{14}$ and make the change of variables
\begin{gather}\label{xalphabeta}\alpha=\frac{\xi^2}{|E|^2}=(2\omega)^2\xi^2,\qquad x=\frac{r}{|E|}=2\omega r, \qquad \beta=\frac{E_1^2}{|E|^2}=l^2\in(0,1),\end{gather}
then it becomes
\begin{gather*}
\alpha^4+\big(4x^3+2\beta x^3-6\beta x^5\big)\alpha^3+\big(3\beta x^4+6x^6-14\beta x^6+3\beta x^8\big)\alpha^2\\
\qquad{} +\big({-}6\beta x^7+4x^9+2\beta x^9\big)\alpha+x^{12}.
\end{gather*}

In Lemma \ref{lem_problem_polynomial} we prove that the above expression is nonzero for all $x>1$, $1<\alpha<\alpha_0$. Hence, equation $F(r,\xi)$ uniquely determines $r$ as a function of $\xi$ for all $\frac{1}{2\omega}<\xi<\xi_0=\frac{1}{2\omega_0}$.

Now it remains only to check that the module of the r.h.s.\ in \eqref{cosfi} is less or equal than 1. In terms of parameters $x$, $\alpha$, $\beta$ \eqref{xalphabeta} formula \eqref{cosfi} becomes
\begin{gather*}\cos\varphi=\frac{-\sqrt{\beta}\big(x^2+x\alpha\big)}{x^3+\alpha}.\end{gather*}
Hence, we need to check that for all $x>1$, $\alpha>1$ we have
\begin{gather*}\sqrt{\beta}\big(x^2+\alpha x\big)\leq x^3+\alpha,\qquad \text{i.e.}, \qquad \alpha\big(\sqrt{\beta} x-1\big)\leq x^2(x-\beta).\end{gather*}
The above inequality is obviously satisfied for $1\leq x\leq\frac{1}{\sqrt{\beta}}$. For $x>\frac{1}{\sqrt{\beta}}$ it becomes
\begin{gather*}\alpha\leq\frac{x^2\big(x-\sqrt{\beta}\big)}{x\sqrt{\beta}-1}.\end{gather*}
Let us find the minimum of the function in the r.h.s.\ of the above formula. The point at which the minimum is attained is among zeros of \begin{gather*}x^2-\frac{\beta+3}{2\sqrt{\beta}}x+1=0,\end{gather*}
i.e., \begin{gather*}x_{\min}=\frac{3+\beta+\sqrt{(1-\beta)(9-\beta)}}{4\sqrt{\beta}},\end{gather*}
and hence \begin{gather*}\alpha_{\max}:= \min\limits_{\frac{1}{\sqrt{\beta}}<x} \! \frac{x^2\big(x-\sqrt{\beta}\big)}{\sqrt{\beta}x-1}=\frac{x_{\min}^3-\sqrt{\beta}x_{\min}}{\sqrt{\beta}x-1}=
\frac{27-18\beta-\beta^2+(9-\beta)\sqrt{(1-\beta)(9-\beta)}}{8\beta\sqrt{\beta}}.\end{gather*}
Let us notice that $x_{\min}$ and $\alpha_{\max}$ coincide with $x_0$ \eqref{x0}, $\alpha_0$ \eqref{alpha_0} from Lemma~\ref{lem_problem_polynomial}. This finishes the proof of the theorem.
\end{proof}

\begin{lem}\label{lem_problem_polynomial} Let $\beta\in(0,1)$. Consider the set $($see Fig.~{\rm \ref{Fig_set_M}a)}
\begin{gather}\label{set_M}
 M= \{(x,\alpha)\colon x>1 , \, \alpha>1 , \, P(x,\alpha)\leq0 \},
\end{gather}
where
\begin{gather*}
P(x,\alpha)\equiv\alpha^4+\big(4x^3+2\beta x^3-6\beta x^5\big)\alpha^3+\big(3\beta x^4+6x^6-14\beta x^6+3\beta x^8\big)\alpha^2 \nonumber\\
\hphantom{P(x,\alpha)\equiv}{} +\big({-}6\beta x^7+4x^9+2\beta x^9\big)\alpha+x^{12}.%\label{ineq_pol}
\end{gather*}
Prove that
\begin{gather}\label{alpha_0}
\min\limits_{(x,\alpha)\in M}\alpha =\alpha_0,\qquad \alpha_0=\frac{27-18\beta-\beta^2+(9-\beta)\sqrt{(1-\beta)(9-\beta)}}{8\beta\sqrt{\beta}},
\end{gather}
and the minimum is attained at a single point $(x_0,\alpha_0)\in M$,
where \begin{gather}\label{x0}
 x_0=\frac{3+\beta+\sqrt{(1-\beta)(9-\beta)}}{4\sqrt{\beta}}.
 \end{gather}
\end{lem}
\begin{proof} Let us transform the inequality in \eqref{set_M} $P(x,\alpha)\leq 0$ in the following way: divide both sides by $x^{12}$, and denote
\begin{gather}\label{z,w}
z=\frac{\alpha}{x^3},\qquad w=z+\frac{1}{z},
\end{gather}
then the inequality $P(x,\alpha)\leq0$ can be rewritten in the following equivalent form
\begin{gather}\label{prom_ineq_1}z^4+\big(4+2\beta -6\beta x^2\big)z^3+\(\frac{3\beta}{x^2}+6-14\beta+3\beta x^2\)z^2+\(\frac{-6\beta}{x^2}+4+2\beta\)z+1\leq0.\!\!\!\end{gather}
In the latter inequality \eqref{prom_ineq_1} we multiply both sides by $x^2$, divide by $z^2$ and then rearrange the terms to obtain a bi-quadratic expression in~$x^2$:
\begin{gather}\label{prom_ineq_2}
3\beta(1-2z)x^4+\big(w^2+(4+2\beta)w+(4-14\beta)\big)x^2+3\beta\(1-\frac{2}{z}\) \leq 0.
\end{gather}

\begin{figure}[t]\centering\small
\begin{tikzpicture}
\node at (0,0) {\includegraphics[width=72mm]{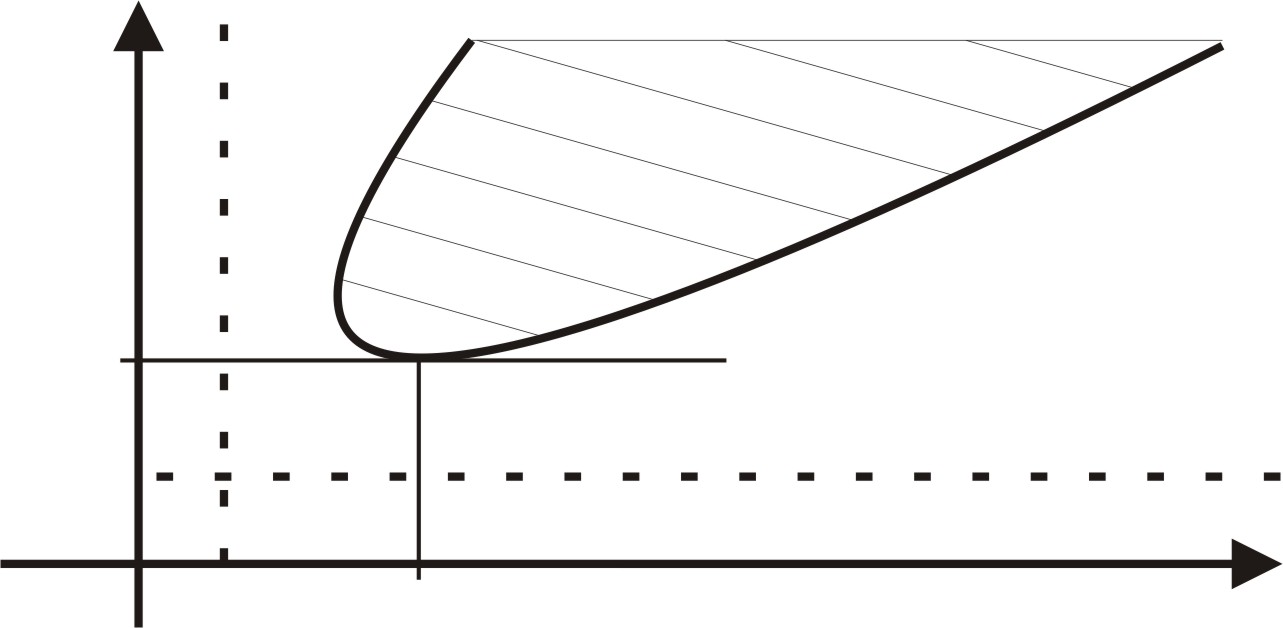}\qquad \includegraphics[width=72mm]{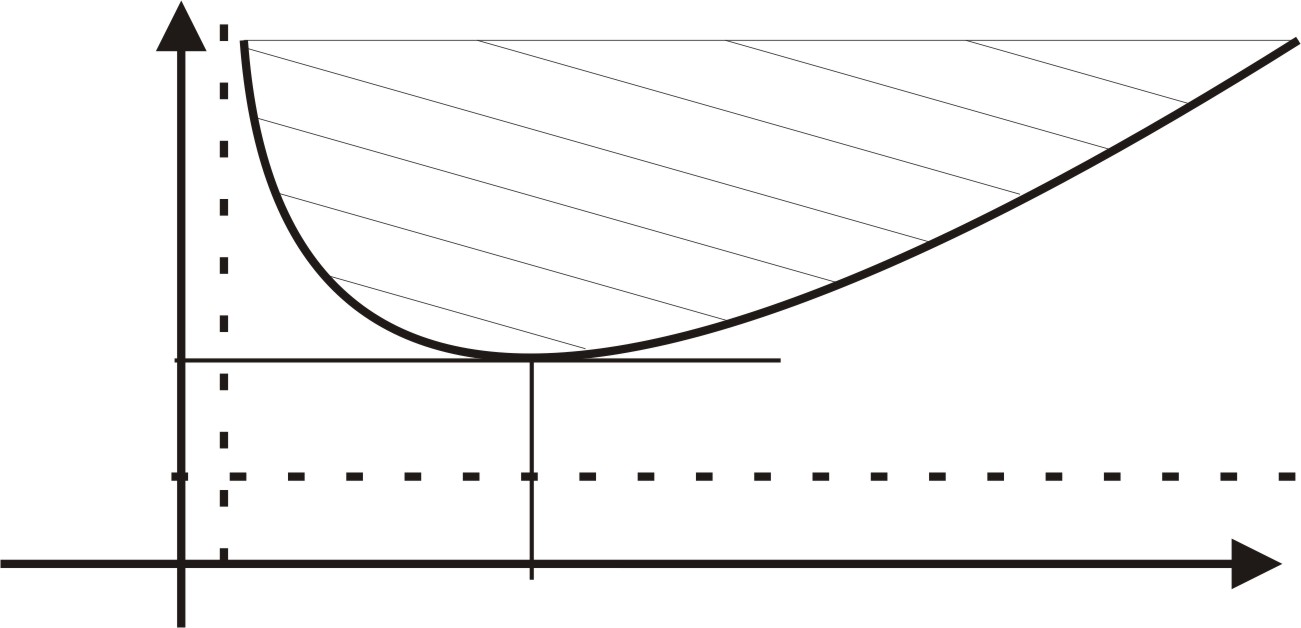}};
\node at (-3.8,-2.2) {(a) Set $M$ \eqref{set_M}};
\node at (4.0,-2.2) {(b) Set $M_1$ \eqref{set_M1}};
\node at (-6.3,-1.7) {$1$};
\node at (-5.15,-1.75) {$x_0$};
\node at (-0.5,-1.75) {$x$};
\node at (1.6,-1.7) {$\frac12$};
\node at (3.4,-1.75) {$z_{_{1,+}}$};
\node at (7.3,-1.75) {$z$};
\node at (-7.0,-0.9) {$1$};
\node at (-7.1,-0.3) {$\alpha_0$};
\node at (-7.1,1.6) {$\alpha$};
\node at (1.15,-0.9) {$1$};
\node at (1.05,-0.3) {$\alpha_0$};
\node at (1.05,1.6) {$\alpha$};
\end{tikzpicture}
\caption{Sets $M$ and $M_1$.}\label{Fig_set_M}
\end{figure}
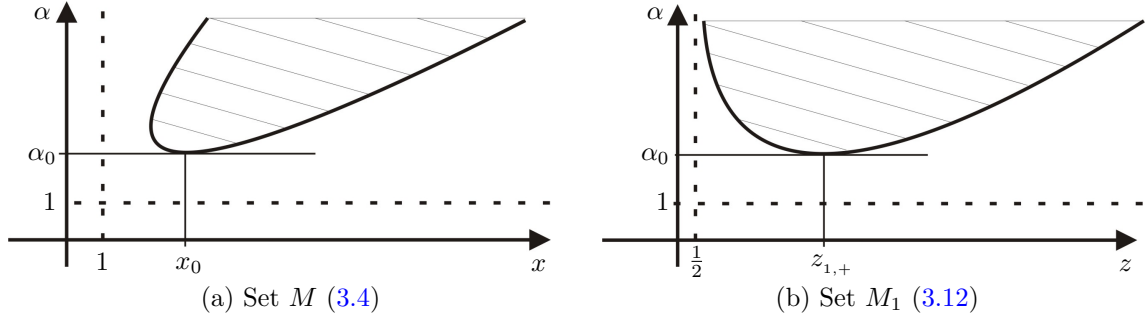

Now, for $(x,\alpha)\in M$ the quantity $z=\frac{\alpha}{x^3}$ can not be equal to $\frac12$. Indeed, for $z=\frac12$ \eqref{prom_ineq_2} becomes \begin{gather*}x^2\leq\frac{\beta}{\frac94-\beta}<\frac45<1,\end{gather*}
which contradicts the condition of $(x,\alpha)\in M$.

The roots of the above bi-quadratic expression in \eqref{prom_ineq_2} are given by the formulas
\begin{gather}\label{roots_x12}
 x_{\pm}^2=\frac{p\pm\sqrt{p^2+36\beta^2(2w-5)}}{6\beta(2z-1)},\qquad p:=(w+2)^2+2\beta(w+2)-18\beta,
\end{gather}
where the index $_{+}$ corresponds to the value of the root with ``$+$'', and the index $_{-}$ corresponds to ``$-$''.

Let us notice, that since $w\geq2$, then
\begin{gather*}p=(w+2)^2+2\beta(w-7)\geq16-10\beta>0,\\
\big((w+2)^2+2\beta(w-7)\big)^2+36\beta^2(2w-5)\geq(16-10\beta)^2-36\beta^2=64(1-\beta)(4-\beta)>0,\end{gather*}
and hence both roots $x_+^2$, $x_-^2$ in \eqref{roots_x12} are always real. By \eqref{z,w}, $z$ can take values in $(0,+\infty)$ when $(x,\alpha)\in M$. Now we need to consider three cases:
\begin{enumerate}\itemsep=0pt
 \item If $0<z<\frac12$. In this case $w>\frac52$, and hence $x_+^2<0$, and $x_-^2>0$. Let us check that $x_-^2<1$. Indeed,
\begin{gather}\label{2w-5}
 2w-5=\frac{(2z-1)(z-2)}{z},
\end{gather}
and hence
\begin{gather*}\sqrt{p^2+36\beta^2(2w-5)}<p+6\beta(1-2z)\end{gather*}
is equivalent in our case to the obvious inequality \begin{gather*}\beta(1-2z)(w+2)(w+2-4\beta)>0.\end{gather*}
Hence, the set described by inequality \eqref{prom_ineq_2}, which reads in this case as
\begin{gather*}x_+^2\leq x^2 \leq x_-^2<1,\end{gather*}
does not intersect with the set $M$ \eqref{set_M}.
\item If $z\geq2$. In this situation $w\geq\frac52$, and hence $x_-^2\leq0$, $x_+^2>0$, hence for $x>1$ \eqref{prom_ineq_2}
reads as \begin{gather*}x^2\geq x_+^2.\end{gather*}
\item If $\frac12<z<2$. In this situation $2\leq w<\frac52$, hence $x_+^2>x_-^2>0$. Let us check that $x_-^2<1$. Indeed,
\begin{gather*}\sqrt{p^2-36\beta^2(5-2w)}>p-6\beta(2z-1)\end{gather*} provided that \begin{gather*}\beta(2z-1)(w+2)(w+2-4\beta)>0,\end{gather*}
and the latter is obviously true. Hence, like in the previous case, \eqref{prom_ineq_2} reads as \begin{gather*}x^2\geq x_+^2.\end{gather*}
\end{enumerate}
Also, for $z>\frac12$ we have $x_+^2>1$. Indeed,
\begin{gather*}\sqrt{p^2+36\beta^2(2w-5)}>-p+6\beta(2z-1)\end{gather*} provided that \begin{gather*}\beta(2z-1)(w+2)(w+2-4\beta)>0,\end{gather*} which is an obvious inequality.

Summing up, under conditions $x>1$, $\alpha>1$, inequality \eqref{prom_ineq_2}, which describes the set~$M$, is equivalent to the inequality \begin{gather*}x^2\geq x_+^2=x_+^2(z),\end{gather*}
with $x_+^2$ determined in \eqref{roots_x12},~\eqref{z,w}.

Hence, in variables $(z,\alpha)$ the set $M$ can be described as follows (see Fig.~\ref{Fig_set_M}b):
\begin{gather}\label{set_M1}M_1=\left\{(z,\alpha)\colon z>\frac12, \, \alpha>1, \, \alpha\geq z\big(x_+^2(z)\big)^{3/2}\right\},\end{gather}
and the map \begin{gather*}(z,\alpha)\mapsto \(x=\(\frac{\alpha}{z}\)^{\frac13},\alpha\)\end{gather*}
is a bijection between sets $M_1$ and $M$. As a byproduct, we see that the set $M$ is not empty, which is not quite clear from the representation~\eqref{set_M}.

Further, the condition $\alpha>1$ in \eqref{set_M1} is reduntant, and the set $M_1$ can be written as
\begin{gather}\label{set_M2}M_1=\left\{(z,\alpha)\colon z>\frac12, \, \alpha\geq z\big(x_+^2(z)\big)^{3/2}\right\}.\end{gather}

Indeed, since $z(x_+(z))^{3/2}>1$ for $z\to\frac12+0$ and for $z\to+\infty$, otherwise we would have a~point $(z,\alpha=1)\in M_1$, and hence a point $(x,\alpha=1)\in M$. Then, for some $x>1$ \begin{gather*}P(x,\alpha=1)\leq0,\end{gather*}
and dividing the latter inequality by $x^{6}$ and making the change of variable $u=x+\frac{1}{x}$, we would obtain
\begin{gather*}u^6-6u^4+2(\beta+2)u^3+3(\beta+3)u^2-12(\beta+1)u+4-20\beta\leq0,\end{gather*}
or, equivalently,
\begin{gather*}(u+2)^2\big((u-1)^4+\beta(2u-5)\big)\leq0,\end{gather*}which can not hold since $u\geq2$ and $\beta\in(0,1)$.

Now, to find $\min\limits_{(x,\alpha)\in M}\alpha$ it is enough to find $\min\limits_{z>\frac12}z\big(x_+^2(z)\big)^{3/2}$. However, the latter problem is itself quite involved. Instead, we notice that the representation~\eqref{set_M2} of the set $M_1$ implies that its boundary
\begin{gather*}\partial M_1=\left\{(z,\alpha)\colon z>\frac12,\, \alpha=z\big(x_+^2(z)\big)^{3/2}\right\}\end{gather*} is ($C^{\infty}$) smooth, and hence the same is true for the boundary of the set~$M$. Hence, all the points $(x,\alpha)$ at which the minimum $\min\limits_{(x,\alpha)\in M}\alpha$ is attained are among the solution set of the system
\begin{gather*}%\label{system_P_Px}
 P(x,\alpha)=0,\qquad P_x(x,\alpha)=0.
\end{gather*}
Equation $P_x(x,\alpha)=0$ for $x>1$ is equivalent to (we divide it by $x^{11}$ and keep the nota\-tions~\eqref{z,w})
\begin{gather*}\big(2+\beta-5\beta x^2\big)z^3+\(\frac{6\beta}{x^2}+6-14\beta+4\beta x^2\)z^2+\(\frac{-7\beta}{x^2}+6+3\beta\)\alpha+2=0,\end{gather*}
or, rearranging it in order to get a biquadratic trinomial in~$x^2$,
\begin{gather}\label{eq_P_x}x^2\beta(-5 z+4)+\((2+\beta)z+6-14\beta+\frac{3(2+\beta)}{z}+\frac{2}{z^2}\)+\frac{\beta\big(2-\frac{7}{z}\big)}{x^2}=0.\end{gather}
On the other hand, equation $P(x,\alpha)=0$ for $x>1$ is equivalent, as have been shown previously, to
\begin{gather}\label{eq_P}x^2=x_+^2=\frac{p+\sqrt{p^2+36\beta^2(2w-5)}}{6\beta(2z-1)}
\end{gather}
with $z>\frac12$ and $z$, $w$ defined in \eqref{z,w} and $p$ defined in \eqref{roots_x12}. It is straightforward to check that $z=2$ is not a solution to~\eqref{eq_P} for $\beta\in(0,1)$. Hence, using \eqref{2w-5}, we can rewrite \eqref{eq_P} as
\begin{gather}\label{x^-2}\frac{1}{x^2}=\frac{z\big(p-\sqrt{p^2+36\beta^2(2w-5)}\big)}{6\beta(2-z)}.\end{gather}
Now we substitute the expressions for $x^2$, $x^{-2}$ from \eqref{eq_P}, \eqref{x^-2} into \eqref{eq_P_x}. Rearranging and collecting terms containing $\sqrt{p^2+36\beta^2(2w-5)}$, we obtain
\begin{gather*}\frac{\sqrt{p^2+36\beta^2(2w-5)} (z+1)^2+3p\big(3z^2-10z+5\big)}{6(2z-1)(z-2)}=(2+\beta)\(z+\frac3z\)+6-14\beta+\frac{2}{z^2}.\end{gather*}
Multiplying by the denominator and transferring terms not containing $\sqrt{p^2+36\beta^2(2w-5)}$ to one side, substituting the expression for~$p$ \eqref{roots_x12}, and extracting then terms $(z+1)^2$, we
obtain
\begin{gather}\label{prom_eq_3}\sqrt{p^2+36\beta^2(2w-5)}=\frac{3\big(1-z^2\big) (3w-2(6-\beta))}{z}.\end{gather}
Further, \begin{gather}\label{prom_eq_4}p^2+36\beta^2(2w-5)=(w+2)^2\big(w^2+4w(1+\beta)+4\big(1-7\beta+\beta^2\big)\big),\end{gather}
and substituting \eqref{prom_eq_4} into \eqref{prom_eq_3}, we obtain
\begin{gather}\label{prom_eq_5}\sqrt{w^2+4w(1+\beta)+4\big(1-7\beta+\beta^2\big)}=\frac{3(1-z) (3w-2(6-\beta) )}{(1+z)}.\end{gather}
Squaring and subtracting, we obtain extra roots corresponding to negative values of the r.h.s.\ of \eqref{prom_eq_5}, which then need to be eliminated. We obtain
\begin{gather*}%\label{prom_eq_6}
 \frac{(z-2)\big(z-\frac12\big)\big(10\big(z^4+1\big)+\big(z^3+z\big)(-77+13\beta)+2z^2\big(75-23\beta+2\beta^2\big)\big)}{z^2}=0,
\end{gather*}
or, dividing by nonzero term $\big(z-\frac12\big)(z-2)$,
\begin{gather*}%\label{prom_eq_7}
10\big(w^2-2\big)+w(-77+13\beta)+2\big(75-23\beta+2\beta^2\big)=0,
\end{gather*}
from where we get two values for $w$: \begin{gather*}%\label{w12}
 w_1=\frac{5-\beta}{2},\qquad \text{or} \qquad w_2=\frac{2(13-2\beta)}{5},
\end{gather*}
each of them in turn produces two values for $z$ in view of \eqref{z,w},
\begin{gather}\label{z12pm}
 z_{1,\pm}=\frac{5-\beta\pm\sqrt{(1-\beta)(9-\beta)}}{4},\qquad z_{2,\pm}=\frac{13-2\beta\pm2\sqrt{(9-\beta)(4-\beta)}}{5}.
\end{gather}
Substituting $w_{1,2}$ into the r.h.s.\ of \eqref{prom_eq_5}, we see that
\begin{gather*}3w_1-2(6-\beta)<0,\qquad \text{and}\qquad 3w_2-2(6-\beta)>0,\end{gather*}
hence for $w_1$ we need to choose $z_{1,+}>1$, while for $w_2$ we need to choose $z_{2,-}\in(0,1)$. However, $z_{2,-}<\frac12$, as can be checked straightforwardly, and we have proven that $z>\frac12$ for $(x,\alpha)\in M$.
Hence, we are left with the only root $z_{1,+}$~\eqref{z12pm}, and by \eqref{eq_P} we find after some computations that
\begin{gather*}p=\frac{-3}{4}\big(\beta^2+18\beta-27\big),\qquad p^2+36\beta^2(2w-5)=\frac{9}{16}(9-\beta)^3(1-\beta),\end{gather*}
\begin{gather*}x^2=\(\frac{\beta+3+\sqrt{(1-\beta)(9-\beta)}}{4\sqrt{\beta}}\)^2,\end{gather*}
and \begin{gather*}\alpha=z x^3=\frac{27-18\beta-\beta^2+(9-\beta)\sqrt{(1-\beta)(9-\beta)}}{8\beta\sqrt{\beta}}.\end{gather*}
We don't need to make an extra check that $x_0>1$, $\alpha_0>1$ (though it is straightforward), since this follows from the previous considerations. This finishes the proof of the lemma.
\end{proof}

\appendix
\section{Appendix}
\begin{teor}[Kotlyarov--Moskovchenko \cite{MoK10}, full formulation]\label{Theor2_full} In the region $\omega_0^2 t<x<\omega^2 t$ the solution of the IBV problem \eqref{SRS}--\eqref{bc} for $t\to\infty$ takes the form of a modulated elliptic wave
\begin{gather*}q(x,t)=2\i\frac{\Theta_{12}(t,\xi;\infty)}{\Theta_{11}(t,\xi;\infty)}\exp\big[2\i t \widehat g_{\infty}(\xi)-2\i\widehat\phi(\xi)\big]+\mathcal{O}\big(t^{-1/2}\big),\\
\nu(x,t)=-1+2\frac{\Theta_{11}(t,\xi;0)\Theta_{22}(t,\xi;0)}{\Theta_{11}(t,\xi;\infty)\Theta_{22}(t,\xi;\infty)}+\mathcal{O}\big(t^{-1/2}\big),\\
\mu(x,t)=2\i\frac{\Theta_{11}(t,\xi;0)\Theta_{12}(t,\xi;0)}{\Theta_{11}^2(t,\xi;\infty)}\exp\big[2\i t \widehat g(\xi)-2\i\widehat\phi(\xi)\big]+\mathcal{O}\big(t^{-1/2}\big),\end{gather*}
where
\begin{gather*}\widehat g_{\infty}(\xi)=\lim\limits_{k\to\infty}\big(\widehat g(k,\xi)-\theta(k,\xi)\big)
=\frac12\(\int_{E}^{\infty}+\int_{\ol{E}}^{\infty}\)\left[\d\widehat g(s,\xi)-\frac{1}{4\xi^2}\right]\d k-\frac{l}{8\omega\xi^2},
\end{gather*}
and $\widehat g$ is defined in \eqref{widehat_g}, \eqref{lambda_-+d_system}. The phase shift $\widehat\phi(\xi)$ is defined by
\begin{gather*}\widehat\phi(\xi)=\frac{1}{2\pi}\int_{\gamma_{d}\cup\ol{\gamma_{d}}}(k-e_1-\zeta_{\infty})\log\big[h(k)\delta^{-2}(k,\xi)\big]\frac{\d k}{w_+(k,\xi)},\end{gather*}
where $e_1=\Re (E+d(\xi))$, and $\zeta_{\infty}(\xi)$ are uniquely determined from the condition that
\begin{gather*}\int_{E}^{k}\frac{z^2-e_1(\xi) z+e_0(\xi)}{w(z)}\d z=k+\zeta_{\infty}(\xi)+\mathcal{O}\big(k^{-1}\big),\qquad k\to\infty,\end{gather*}
and \begin{gather*}\delta(k,\xi)=\exp\left\{\frac{1}{2\pi\i}\int_{\lambda_-(\xi)}^{\lambda_+(\xi)}\frac{\log\big(1-\rho^2(s)\big)\d s}{s-k}\right\},\qquad k\in\mathbb{C}\setminus[\lambda_-(\xi), \lambda_+(\xi)],\\
h(k)=\begin{cases}-\i f(k),& k\in\gamma_{d},\\ \i f^{-1}(k), & k\in\ol{\gamma}_d,\end{cases}\end{gather*}
where $f$, $\rho$ are defined in \eqref{rho_def}, \eqref{f_def}. Further,
\begin{gather*}%\label{Theta}
\Theta_{11}(t,\xi;k)=\frac12\(\widetilde\varkappa(k)+\frac{1}{\widetilde\varkappa(k)}\)
\frac{\theta_3\left[U(k)-U(E_0)-\frac{\tau}{2}-\frac{t B_g+B_{\zeta}\Delta}{2\pi}\right]}{\theta_3\left[U(k)-U(E_0)-\frac12-\frac{\tau}{2}\right]},\\
\Theta_{12}(t,\xi;k)=\frac12\(\widetilde\varkappa(k)-\frac{1}{\widetilde\varkappa(k)}\)
\frac{\theta_3\left[U(k)+U(E_0)+\frac{\tau}{2}+\frac{t B_g+B_{\zeta}\Delta}{2\pi}\right]}{\theta_3\left[U(k)+U(E_0)+\frac12+\frac{\tau}{2}\right]},\\
\Theta_{21}(t,\xi;k)=\frac12\(\widetilde\varkappa(k)-\frac{1}{\widetilde\varkappa(k)}\)
\frac{\theta_3\left[U(k)+U(E_0)+\frac{\tau}{2}-\frac{t B_g+B_{\zeta}\Delta}{2\pi}\right]}{\theta_3\left[U(k)+U(E_0)+\frac12+\frac{\tau}{2}\right]},\\
\Theta_{22}(t,\xi;k)=\frac12\(\widetilde\varkappa(k)+\frac{1}{\widetilde\varkappa(k)}\)
\frac{\theta_3\left[U(k)-U(E_0)-\frac{\tau}{2}+\frac{t B_g+B_{\zeta}\Delta}{2\pi}\right]}{\theta_3\left[U(k)-U(E_0)-\frac12-\frac{\tau}{2}\right]},
\end{gather*}
where
\begin{gather*}U(k)=\(2\int^{d}_{\ol{d}}\frac{\d z}{w(z)}\)^{-1}\int_{E}^{k}\frac{\d z}{w(z)},\qquad \tau=\(\int^{d}_{\ol{d}}\frac{\d z}{w(z)}\)^{-1}\int_{E}^{d}\frac{\d z}{w(z)},\\
E_0=\frac{E_1 d_2+E_2 d_1}{E_2+d_2}, \qquad \widetilde\varkappa(k)=\(\frac{(k-\ol{E})(k-\ol{d(\xi)})}{(k-E)(k-d(\xi))}\)^{1/4},\\
\theta_3(z)=\sum\limits_{m\in\mathbb{Z}}\e^{\pi\i\tau m^2+2\pi\i m z},\qquad B_{g}=\frac12\left[\int_{E}^{d}+\int_{\ol E}^{\ol d}\right]\d \widehat g(z,\xi)\in\mathbb{R},\\
\Delta=\frac{1}{2\pi}\int_{\gamma_d+\ol\gamma_d}\frac{\log\big[h(s)\delta^{-2}(s,\xi)\big]}{w_+(s,\xi)}\in\mathbb{R}, \qquad B_{\zeta}=2\int_{E}^{d}\frac{\big(z^2-e_1 z+e_0\big)\d z}{w(z)}\in\mathbb{R}.\end{gather*}
\end{teor}

\subsection*{Acknowledgements} A.M.~would like to thank Vladimir Kotlyarov for useful discussions, and Koen van den Dungen and R\'{e}my Rodiac for careful reading a version of this manuscript and giving valuable comments and suggestions. Also the authors thank the anonymous referees for careful reading of the manuscript and their comments and suggestions. Last, but not least, A.M.~is also grateful to Questura di Trieste, in a queue to which a part of this work was done.

\pdfbookmark[1]{References}{ref}
\LastPageEnding

\end{document}